\documentclass[runningheads]{llncs}
\usepackage{lipsum}

\usepackage{floatflt, float}
\usepackage{url}

\usepackage{wrapfig}

\usepackage{xcolor}
\usepackage{amsmath}
\usepackage{amssymb, amsbsy}
\usepackage{array}
\usepackage{cite}
\usepackage{wrapfig}
\usepackage{multicol}
\usepackage{graphicx}
\graphicspath{ {images/} }
\usepackage{epstopdf}
\usepackage{color, import}
\usepackage{sectsty}
\usepackage{enumitem}
\usepackage{lineno,hyperref}
\usepackage[linesnumbered,vlined,ruled]{algorithm2e}
\usepackage{subcaption}
\SetAlFnt{\small}

\usepackage{setspace}

\usepackage[list=true]{subcaption}

\newcolumntype{P}[1]{>{\centering\arraybackslash}p{#1}}
\newcolumntype{M}[1]{>{\centering\arraybackslash}m{#1}}

\begin{document}

\title{The Computational Landscape of Autonomous Mobile Robots: The Visibility Perspective  }

\author{Archak Das\orcidID{0000-0002-1630-3052} \and Satakshi Ghosh\orcidID{0000-0003-1747-4037}\and
Avisek Sharma\orcidID{0000-0001-8940-392X}\and
Pritam Goswami\orcidID{0000-0002-0546-3894}\and
Buddhadeb Sau\orcidID{0000-0001-7008-6135} }

\authorrunning{A. Das, S. Ghosh, A. Sharma, P. Goswami, and B. Sau}

\institute{Department of Mathematics, Jadavpur University, Kolkata, India \\
\email{\{archakdas.math.rs, satakshighosh.math.rs, aviseks.math.rs, pritamgoswami.math.rs,  buddhadeb.sau\}@jadavpuruniversity.in}
}

\maketitle              

\begin{abstract}
Consider a group of autonomous mobile computational entities called robots. The robots move in the Euclidean plane and operate according to synchronous $Look$-$Compute$-$Move$ cycles. The computational capabilities of the robots under the four traditional models $\{ \mathcal{OBLOT},\  \mathcal{FSTA},\  \mathcal{FCOM},\  \mathcal{LUMI} \} $ have been extensively investigated both when the robots had unlimited amount of energy and when the robots were energy-constrained. 

In both the above cases, the robots had full visibility. In this paper, this assumption is removed, i.e., we assume that the robots can view up to a constant radius $V_r$ from their position (the $V_r$ is same for all the robots) and, investigates what impact it has on its computational capabilities.

We first study whether the restriction imposed on the visibility has any impact at all, i.e., under a given model and scheduler does there exist any problem which cannot be solved by a robot having limited visibility but can be solved by a robot with full visibility. We find that the answer to the question in general turns out to be positive. Finally, we try to get an idea that under a given model, which of the two factors, $Visibility$ or $Synchronicity$ is more powerful and conclude that a definite conclusion cannot be drawn.


\end{abstract}
\section{Introduction}
In this paper, we consider systems of $autonomous$, $anonymous$, $identical$ and $homogenous$ computational entities called $robots$ moving and operating in the Euclidean plane. The robots are viewed as points and they operate according to the traditional $Look$-$Compute$-$Move$ ($LCM$) cycles in synchronous rounds. In $Look$ phase robots take a snapshot of the space, in the $Compute$ phase the robots execute its algorithm using the snapshot as input, then move to the computed destination in $Move$ phase. The robots are collectively able to perform some tasks and solve some problems.


In recent times, exhaustive investigation has also been done \cite{FlocchiniSW19, BuchinFKPSW21} about the issues of $memory$ $persistence$ and $communication$ $capability$, have on the solvability of a problem and computational capability of a system. In light of these facts, four models have been identified and investigated, $\mathcal{OBLOT}$, $\mathcal{LUMI}$, $\mathcal{FSTA}$ and $\mathcal{FCOM}$. 

In the most common model $\mathcal{OBLOT}$, in the addition to standard assumptions of $anonymity$ and $uniformity$, the robots are $silent$, i.e., without explicit means of communication, and $oblivious$, i.e., they have no persistent memory to record information of previous cycles.

The other model which is generally considered as antithesis to $\mathcal{OBLOT}$ model, is the $\mathcal{LUMI}$ model first formally defined in \cite{DasFPSY12, 0001FPSY16}, where the robots have both persistent memory and communication, although it must be noted that the remembering or communicating can be done only in limited capacity, i.e., the robots can remember or communicate finite number of bits. 

Two new models have been introduced in \cite{FlocchiniSVY16} by eliminating one of the two capabilities of $\mathcal{LUMI}$ model, in each case. These two models are $\mathcal{FSTA}$ and $\mathcal{FCOM}$. In $\mathcal{FSTA}$ model the communication capability is absent while in $\mathcal{FCOM}$ model the robots do not have persistent memory. In \cite{FlocchiniSW19} these models have been considered to investigate the question $is$  $it$ $better$ $to$ $remember$ $or$ $to$ $communicate?$.



In this work we consider another factor, i.e., $visibility$, which helps to investigate the matter from a different angle and is of interest from both theoretical and practical point of view. If $\mathcal{M}$ denotes a model and $\sigma$ is a scheduler then traditionally $\mathcal{M}^{\sigma}$ denotes a model $\mathcal{M}$ under scheduler $\sigma$. Here $\mathcal{M} \in \{ \mathcal{OBLOT}, \mathcal{FSTA}, \mathcal{FCOM}, \mathcal{LUMI} \}$ and $\sigma \in \{ FSYNCH, SSYNCH \}$. 
We here define $\mathcal{M_{V}^{\sigma}}$ denotes a model $\mathcal{M}$ under scheduler $\sigma$ and visibility state $\mathcal{V}$. Here $\mathcal{V} \in \{\mathcal{F.V.}, \mathcal{L.V.} \}$. Here $\mathcal{F.V.}$ denotes the full visibility model and $\mathcal{L.V.}$ denotes the limited visibility model. In limited visibility model, each robot can view upto a constant radius of their position  and the initial visibility graph is connected.

All the works done till now had given the full visibility power to the robots. So we try to answer a number of questions, 
\begin{itemize}
    \item Does restriction on visibility have a significant impact on the computational power of a model? In other words, is any problem which is solvable in full visibility model, solvable in limited visibility model also?
    \item If the answer to the second question mentioned above is no, then can the impairment be $always$ adjusted by minor adjustments, i.e., by keeping the model intact but by making the scheduler more powerful?
    \item If the answer to the above question is also no, then what are the cases where the lack of full visibility can be compensated? Can all the cases be classified? 
\end{itemize}

We have shown the answer to the first question is no, e.g., we have proved that $\mathcal{FSTA_{L.V.}^{F}} < \mathcal{FSTA_{F.V.}^{F}}$. This result turns out to be true for all the four models $\{ \mathcal{OBLOT},\ \mathcal{FSTA},\  \mathcal{FCOM}, \mathcal{LUMI} \}$. After we got answer to our first question, we tried to answer our second question.  A definite answer to the second question shall give us some insight about whether any one of the two parameters of $visibility$ and $synchronicity$ have any precedence over the other. But as we shall see this does not happen. For e.g.,  we got the result that $\mathcal{FCOM_{L.V.}^{F}} \perp \mathcal{FCOM_{F.V.}^{S}}$, which effectively shows that both the capabilities are important, and deficiency in one of the parameter cannot be compensated by making the other parameter stronger. In the process, we have defined a new problem $EqOsc$ which we have shown to be unsolvable unless the scheduler is fully synchronous, but solvable if the scheduler is fully synchronous, even under limited visibility in $\mathcal{FSTA}$ and $\mathcal{FCOM}$ models. The third question as of now, yet remains unanswered, and subject to further investigations.

The results presented in this paper gives a whole new insight to the parameter of visibility and its relevance relative to the parameter of synchronicity, one that requires exhaustive analysis, even beyond the amount of investigation that had been done in this paper.

\subsection{Related Work}

Investigations regarding the computational power of robots under synchronous schedulers was done by the authors Flocchini et. al. in \cite{FlocchiniSW19}. Main focus of the investigation in this work was which of the two capabilities was more important: $persistent$ $memory$ or $communication$. In the course of their investigation  they proved that under fully synchronous scheduler communication is more powerful than persistent memory. In addition to that, they gave a complete exhaustive map of the relationships among models and schedulers.

In \cite{BuchinFKPSW21}, the previous work of characterizing the relations between the robot models and three type of schedulers was continued. The authors provided a more refined map of the computational landscape for robots operating under fully synchronous and semi-synchronous schedulers, by removing the assumptions on robots' movements of $rigidity$ and common $chirality$. Further authors establish some preliminary results with respect to asynchronous scheduler.

The previous two works considered that the robots was assumed to have unlimited amounts of energy. In \cite{BuchinFKPSW22}, the authors removed this assumption and started the study of computational capabilities of robots whose energy is limited, albeit renewable. In these systems, the activated entities uses all its energy to  execute an $LCM$ cycle and then the energy is restored after a period of inactivity. They studied the impact that memory persistence and communication capabilities have on the computational power of such robots by analyzing the computational relationship between the four models  $\{ \mathcal{OBLOT}, \mathcal{FSTA}, \mathcal{FCOM}, \mathcal{LUMI} \} $ under the energy constraint. They provided a full characterization of this relationship. Among the many results they proved that for energy-constrained robots, $\mathcal{FCOM}$ is more powerful than $\mathcal{FSTA}$.

In all the three above mentioned works, the robots had full visibility. A robot uses its visibility power in the $Look$ phase of the $LCM$ cycle to acquire information about its surroundings, i.e., position and lights (if any) of other robots. The biggest drawback of full visibility assumption is that it is not practically feasible. So, recently some of the authors \cite{FlocchiniPSW05, AndoOSY99, PoudelS21, GoswamiSGS22} have considered limited visibility. For example, in \cite{FlocchiniPSW05} they considered that the robots can see up to a fixed radius $V$ from it. So it was important to study the power of different robot models under limited visibility.

\subsection{Our Contributions}
Let $\mathcal{M^X_V}$ denote a model $\mathcal{M}$ under scheduler $\mathcal{X}$ with visibility $\mathcal{V}$. Here  $\mathcal{M} \in \{ \mathcal{OBLOT},\ \mathcal{FSTA},\  \mathcal{FCOM}, \mathcal{LUMI} \}$,  $\mathcal{X} \in \{\mathcal{F}, \mathcal{S}\}$, here $\mathcal{F}$ and $\mathcal{S}$ denotes $FSYNCH$ and $SSYNCH$ schedulers respectively, and $\mathcal{V} \in \{\mathcal{F.V.}, \mathcal{L.V.} \}$. $\mathcal{A^B_C} \geq \mathcal{D^E_F}$ denotes that the model $\mathcal{A}$, under scheduler $\mathcal{B}$ and visibility $\mathcal{C}$ is computationally not weaker than  the model $\mathcal{D}$ under scheduler $\mathcal{E}$ and visibility $\mathcal{F}$, $\mathcal{A^B_C} > \mathcal{D^E_F}$ denotes that the model $\mathcal{A}$ under scheduler $\mathcal{B}$ and visibility $\mathcal{C}$ is computationally more powerful than  the model $\mathcal{D}$ under scheduler $\mathcal{E}$ and visibility $\mathcal{F}$, $\mathcal{A^B_C} \equiv \mathcal{D^E_F}$ denotes that the model $\mathcal{A}$ under scheduler $\mathcal{B}$ and visibility $\mathcal{C}$ is computationally equivalent to  the model $\mathcal{D}$ under scheduler $\mathcal{E}$ and visibility $\mathcal{F}$ and $\mathcal{A^B_C} \perp \mathcal{D^E_F}$ denotes that the model $\mathcal{A}$ under scheduler $\mathcal{B}$ and visibility $\mathcal{C}$ is computationally incomparable with  the model $\mathcal{D}$ under scheduler $\mathcal{E}$ and visibility $\mathcal{F}$.

We first examine the computational relationship under a constant model and scheduler between limited and full visibility conditions. We find that in both fully synchronous and semi-synchronous cases, a model under full visibility is strictly more powerful than a model under limited visibility. We get the following results:

\begin{enumerate}
    \item $\mathcal{OBLOT^F_{F.V.}> OBLOT^F_{L.V.}}$
    \item $\mathcal{FSTA^F_{F.V.}> FSTA^F_{L.V.}}$
    \item $\mathcal{FCOM^F_{F.V.}> FCOM^F_{L.V.}}$
    \item $\mathcal{LUMI^F_{F.V.}> LUMI^F_{L.V.}}$
    \item $\mathcal{OBLOT^S_{F.V.}> OBLOT^S_{L.V.}}$
    \item $\mathcal{FSTA^S_{F.V.}> FSTA^S_{L.V.}}$
    \item $\mathcal{FCOM^S_{F.V.}> FCOM^S_{L.V.}}$
    \item $\mathcal{LUMI^S_{F.V.}> LUMI^S_{L.V.}}$
\end{enumerate}

We then examine the computational relationship between the four models $\{ \mathcal{OBLOT},\mathcal{FSTA},\mathcal{FCOM}, \mathcal{LUMI} \}$ under limited visibility between fully synchronous and semi-synchronous schedulers. We find that under limited visibility conditions each of the three models are more powerful under fully synchronous scheduler.

\begin{enumerate}
 \setcounter{enumi}{8}
    \item $\mathcal{OBLOT^F_{L.V.}> OBLOT^S_{L.V.}}$
    \item $\mathcal{FSTA^F_{L.V.}> FSTA^S_{L.V.}}$
    \item $\mathcal{FCOM^F_{L.V.}> FCOM^S_{L.V.}}$
    \item $\mathcal{LUMI^F_{L.V.}> LUMI^S_{L.V.}}$
    
\end{enumerate}

Together with the three results mentioned immediately above, we also get an idea which of the capabilities, $visibility$ or $synchronicity$ is more powerful. From the previous results we generally conclude that $\mathcal{M^S_{L.V.}} < \mathcal{M^S_{F.V.}}$ and $\mathcal{M^S_{L.V.}} < \mathcal{M^F_{L.V.}}$. If we can prove that $\mathcal{M^F_{L.V.}} \geq \mathcal{M^S_{F.V.}}$, then it shall imply that by making the scheduler stronger, the limitation in visibility can be compensated. Similarly if we can prove that $\mathcal{M^S_{F.V.}} \geq \mathcal{M^F_{L.V.}}$, then it shall imply that by giving complete visibility, the weakness in terms of scheduler can be compensated. But after our detailed investigation it has been revealed that neither of the above cases happen and in general $\mathcal{M^F_{L.V.}} \perp \mathcal{M^S_{F.V.}}$. Specifically the results are:

\begin{enumerate}
 \setcounter{enumi}{12}
    \item $\mathcal{OBLOT^F_{L.V.} \perp OBLOT^S_{F.V.}}$
    \item $\mathcal{FSTA^F_{L.V.} \perp FSTA^S_{F.V.}}$
    \item $\mathcal{FCOM^F_{L.V.} \perp FCOM^S_{F.V.}}$
    \item $\mathcal{LUMI^F_{L.V.} \perp LUMI^S_{F.V.}}$
    
\end{enumerate}

\subsection{Paper Organization}
 In Section \ref{2} we define the computational model, visibility models and the other technical preliminaries. In Sections \ref{main sec}, we discuss the computational relationships between the four models, $\mathcal{OBLOT}$, $\mathcal{FCOM}$, $\mathcal{FSTA}$ and $\mathcal{LUMI}$ models respectively, subject to variations in synchronicity and visibility. In Section \ref{6}, we try to compare the strength of the capabilities of $synchronicity$ and $visibility$. In Section \ref{7} we present the conclusion.

\section{Model and Technical Preliminaries}\label{2}
\subsection{The Basics}
In this paper we consider a team $R = \{ r_0, \ldots, r_n \}$ of computational entities moving and operating in the Euclidean Plane $\mathbb{R}^2$, which are viewed as points and called $robots$. The robots can move freely and continuously in the plane. Each robot has its own local coordinate system and it always perceives itself at its origin; there might not be consistency between these coordinate systems. The robots are $identical$: they are indistinguishable by their appearance and they execute the same protocol, and they are $autonomous$, i.e., without any central control.

The robots operate in $Look-Compute-Move$ $(LCM)$ cycles. When activated a robot executes a cycle by performing the following three operations:

\begin{enumerate}
    \item $Look$: The robots activate its sensors to obtain a snapshot of the positions occupied by the robots according to its co-ordinate system.
    \item $Compute$: The robot executes its algorithm using the snapshot as input. The result of the computation is a destination point.
    \item $Move$: The robot moves to the computed destination. If the destination is the current location, the robot stays still.
\end{enumerate}

All robots are initially idle. The amount of time to complete a cycle is assumed to be finite, and the $Look$ is assumed to be instantaneous.
The robots may not have any agreement in terms of their local coordinate system. By $chirality$, we mean the robots agree on the same circular orientation of the plane, or in other words they agree on "clockwise" direction. In our paper, we do not assume the robots to have a common sense of chirality.

\subsection{The Computational Models}

There are four basic robot models which are considered in literature, they are namely, $\{ \mathcal{OBLOT}, \mathcal{FSTA}, \mathcal{FCOM}, \mathcal{LUMI} \}$.

In the most common, $\mathcal{OBLOT}$, the robots are $silent$: they have no explicit means of communication; furthermore they are $oblivious$: at the start of the cycle, a robot has no memory of observations and computations performed in previous cycles.

In the most common model, $\mathcal{LUMI}$, each robot $r$ is equipped with a persistent visible state variable $Light[r]$, called $light$, whose values are taken from a finite set $C$ of states called $colors$ (including the color that represents the initial state when the light is off). The colors of the lights can be set in each cycle by $r$ at the end of its $Compute$ operation. A light is $persistent$ from one computational cycle to the next: the color is not automatically reset at the end of a cycle; the robots are otherwise oblivious, forgetting all other information from previous cycles. In $\mathcal{LUMI}$, the $Look$ operation produces a colored snapshot; i.e., it returns the set of pairs ($position$, $color$) of the robots. Note that if $\lvert C \rvert =1$, then the light is not used; this case corresponds to the $\mathcal{OBLOT}$ model.

It is sometimes convenient to describe a robot $r$ as having $k \geq 1$ lights, denoted $r.light_1, \ldots, r.light_k$, where the values of $r.light_i$ are from a finite set of colors $C_i$, and to consider $Light[r]$ as a $k$-tuple of variables; clearly, this corresponds to $r$ having a single light that uses $\prod_{i=1}^{k} \lvert C_i \rvert$ colors.

The lights provide simultaneously persistent memory and direct means of communication, although both limited to a constant number of bits per cycle. Two sub-models of $\mathcal{LUMI}$ have been defined and investigated , each offering only one of these two capabilities.

In the first model, $\mathcal{FSTA}$, a robot can only see the color of its own light; that is, the light is an $internal$ one and its color merely encodes an internal state. Hence the robots are $silent$, as in $\mathcal{OBLOT}$, but are $finite$-$state$. Observe that a snapshot in $\mathcal{FSTA}$ is same as in $\mathcal{OBLOT}$.

In the second model, $\mathcal{FCOM}$, the lights $external$: a robot can communicate to the other robots through its colored light but forgets the color of its own light by the next cycle; that is, robots are $finite$-$communication$ but are $oblivious$. A snapshot in $\mathcal{FCOM}$ is like in $\mathcal{LUMI}$ except that, for the position $x$ where the robot $r$ performing the $Look$ is located, $Light[r]$ is omitted from the set of colors present at $x$.

In all the above models, a $configuration$ $C(t)$ at time $t$ is the multi-set of the $n$ pairs of the ($x $), where $c_{i}$ is the color of robot $r_i$ at time $t$.

\subsection{The Schedulers}

With respect to the activation schedule of the robots, and the duration of their $Look$-$Compute$-$Move$ cycles, the fundamental distinction is between the $asynchronous$ and $synchronous$ settings.

In the $asynchronous$ setting (ASYNCH), there is no common notion of time, each robot is activated independently of others, the duration of each phase is finite but unpredictable and might be different cycles.

In the $synchronous$ setting (SSYNCH), also called semi-synchronous, time is divided into discrete intervals, called $rounds$; in each round some robots are activated simultaneously, and perform their $LCM$ cycle in perfect synchronization.

A popular synchronous setting which plays an important role is the $fully$-$synchronous$ setting (FSYNCH) where every robot is activated in every round; the is, the activation scheduler has no adversarial power.

In all two settings, the selection of which robots are activated at a round is made by an adversarial $scheduler$, whose only limit is that every robot must be activated infinitely often (i.e., it is fair scheduler). In the following, for all synchronous schedulers, we use round and time interchangeably.

\subsection{The Visibility Models}
In our work we do comparative analysis of computational models with robots having $full$ and $limited$ visibility. In $full$ $visibility$ model, denoted as $\mathcal{F.V.}$, the robots have sensorial devices that allows it to observe the positions of the other robots in its local co-ordinate system.

In $limited$ $visibility$ model, denoted as $\mathcal{L.V.}$, a robot can only observe upto a fixed distance $V_r$ from it. Suppose, $r_p(t)$ denote the position of a robot $r$ at the beginning of round $t$. Then we define the circle with center at $r_p(t)$ and radius $V_r$ to be the $Circle$  $of$ $Visibility$ of $r$ at round $t$. Here the radius $V_r$ is same for all the robots. The result of $Look$ operation in round $t$ will be the position of the robots and lights(if any) of the robots inside the circle of visibility.

We now define the $Visibility$ $Graph$, $G =(V,E)$ of a configuration. Let $C$ be a given configuration. Then all the robot positions become the vertices of $G$ and we say an edge exists between any two vertices if and only if the robots present there can see each other. The necessary condition for the problems we have defined in the paper is that the $Visibility$ $Graph$ of the initial configuration must be connected.

\subsection{Some Important Definitions}
We define our computational relationships similar to that of \cite{FlocchiniSW19}. Let $\mathcal{M}= \{ \mathcal{OBLOT}, \mathcal{FSTA}, \mathcal{FCOM}, \mathcal{LUMI}  \}$ be the robot models under investigation, the set of activation schedulers be $\mathcal{S} = \{ FSYNCH, ASYNCH\}$ and the set of visibility models be $ \mathcal{V} = \{\mathcal{F.V.}, \mathcal{L.V.} \}$.

We denote by $\mathcal{R}$ the set of all teams of robots satisfying the core assumptions (i.e., they are identical, autonomous, and operate in $LCM$ cycles), and $R \in \mathcal{R}$ a team of robots having identical capabilities (e.g., common coordinate system, persistent storage, internal identity, rigid movements etc.). By $\mathcal{R}_n \subset R$ we denote the set of all teams of size $n$.

By \textit{problem} we mean a task where a fixed number of robots have to form some configuration or configurations (which may be a function of time) subject to some conditions, within a finite amount of time.

Given a model $M \in \mathcal{M}$, a scheduler $S \in \mathcal{S}$, visibility $V \in \mathcal{V}$, and a team of robots $R \in \mathcal{R}$, let $Task(M,S,V;R)$ denote the set of problems solvable by $R$ in $M$, with visibility $V$ and under adversarial scheduler $S$.

Let $M$, $N$ $\in$ $\mathcal{M}$, $S_1$, $S_2$ $\in$ $\mathcal{S}$ and $V_1$, $V_2$ $\in$ $\mathcal{V}$. We define the relationships between model $M$ with visibility $V_1$ under scheduler $S_1$ and model $N$ with visibility $V_2$ under scheduler $S_2$:

\begin{itemize}
    \item $computationally$ $not$ $less$ $powerful$ ($M^{S_1}_{V_1}$ $\geq$ $N^{S_2}_{V_2}$), if $\forall$ $R \in \mathcal{R}$ we have $Task(M, S_1;R)$ $\supseteq$ $Task(N, S_2;R)$;

    \item $computationally$ $more$ $powerful$ ($M^{S_1}_{V_1}$ $>$ $N^{S_2}_{V_2}$), if $M^{S_1}_{V_1}$ $\geq$ $N^{S_2}_{V_2}$ and $\exists R \in \mathcal{R}$ such that $Task(M, S_1, V_1;R)$ $\setminus$ $Task(N, S_2, V_2;R)$ $\neq$ $\emptyset$;
    
    \item $computationally$ $equivalent$ ($M^{S_1}_{V_1}$ $\equiv$ $N^{S_2}_{V_2}$), if $M^{S_1}_{V_1}$ $\geq$ $N^{S_2}_{V_2}$ and $M^{S_1}_{V_1}$ $\leq$ $N^{S_2}_{V_2}$;

    \item $computationally$ $orthogonal$ $or$ $incomparable$, ($M^{S_1}_{V_1}$ $\perp$ $N^{S_2}_{V_2}$), if $\exists R_1, R_2 \in \mathcal{R}$ such that $Task(M, S_1, V_1;R_1)$ $\setminus$ $Task(N, S_2, V_2;R_1)$ $\neq$ $\emptyset$ and $Task(N, S_2, V_2;R_2)$ $\setminus$ $Task(M, S_1, V_1;R_2)$ $\neq$ $\emptyset$.
\end{itemize}

For simplicity of notation, for a model $M$ $\in$ $\mathcal{M}$, let $M^{F}$ and $M^{S}$ denote $M^{Fsynch}$ and $M^{Ssynch}$, respectively; and let $M^{F}_{V}(R)$ and $M^{S}_{V}(R)$ denote the sets $Task(M, FSYNCH, V;R)$ and $Task(M, SSYNCH, V;R)$, respectively.

\subsection{Some Fundamental Comparisons}

Trivially, 
\begin{enumerate}
    \item $\mathcal{LUMI} \geq \mathcal{FSTA} \geq \mathcal{OBLOT}$  and $\mathcal{LUMI} \geq \mathcal{FCOM} \geq \mathcal{OBLOT}$, when the $Visibility$ and $Synchronicity$ is fixed.
    \item $\mathcal{FSYNCH} \geq \mathcal{SSYNCH} \geq \mathcal{ASYNCH}$   when the model and $Visibility$ is fixed.
    \item $\mathcal{F.V.} \geq \mathcal{L.V.}$   when the model and $Synchronicity$ is fixed.

\end{enumerate}

\section{ANGLE EQUALIZATION PROBLEM}\label{main sec}
\begin{definition}
    \textbf{Problem Angle Equalization (AE)}: Suppose four robots $r_1$, $r_2$, $r_3$ and $r_4$ are placed in positions $A$, $B$, $C$ and $D$ respectively, as given in Configuration (I). The line $AB$ makes an acute angle $\theta_1$ with $BC$ and the line $CD$ makes an acute $\theta_2$ with $BC$. Here $\theta_1 < \theta_2 < 90^\circ $.

    The robots must form the Configuration (II) without any collision. The robots $r_2$ and $r_3$ must remain fixed in their positions.
\end{definition}

\subsection{Algorithm for $AE$ problem in $\mathcal{OBLOT_{F.V}^S}$}

\begin{figure}[h!]
\begin{minipage}[H]{0.45\linewidth}
\centering 
\includegraphics[width=3cm]{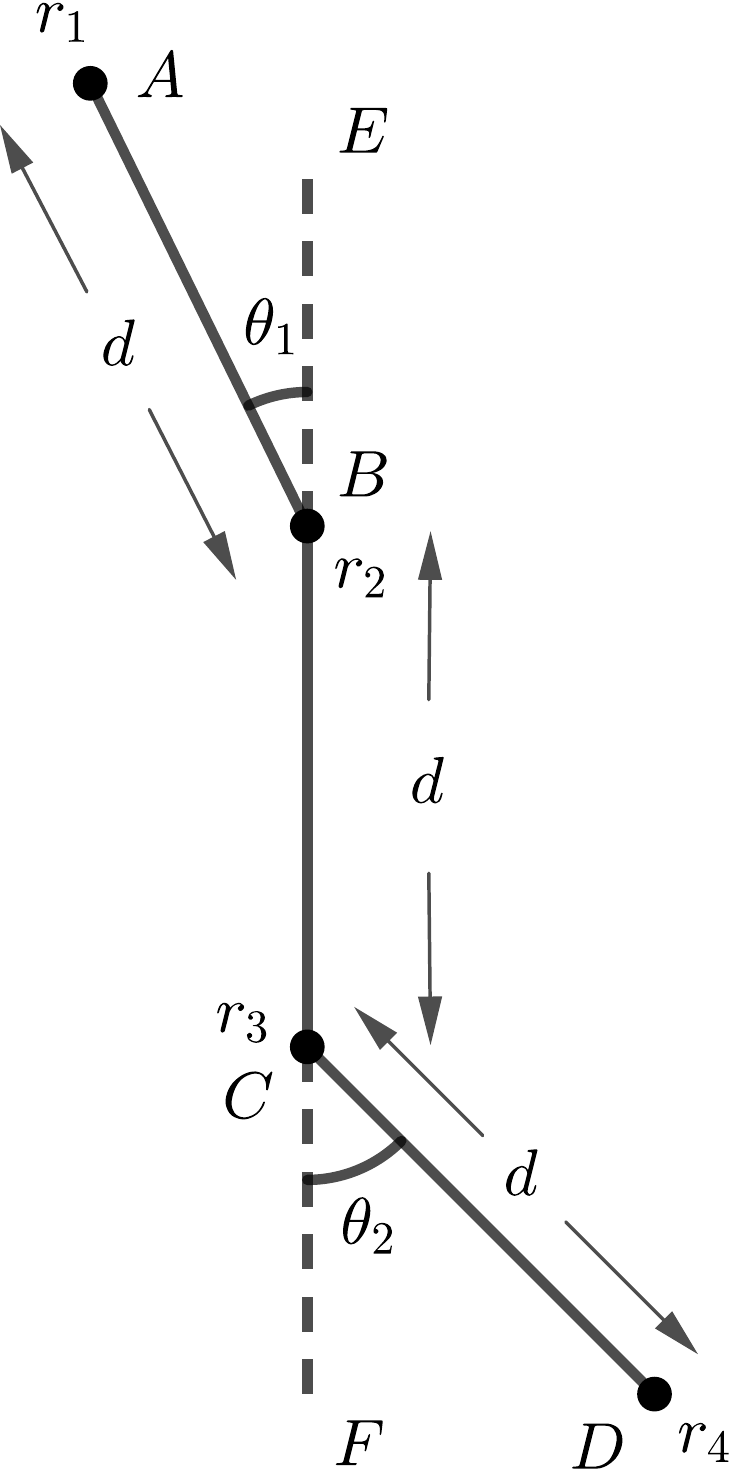}
\caption{Configuration (I) of Problem $AE$}
\label{fig:Iso1}
\end{minipage}
\hfill
\begin{minipage}[H]{0.45\linewidth}
 \centering   
\includegraphics[width=3.5cm]{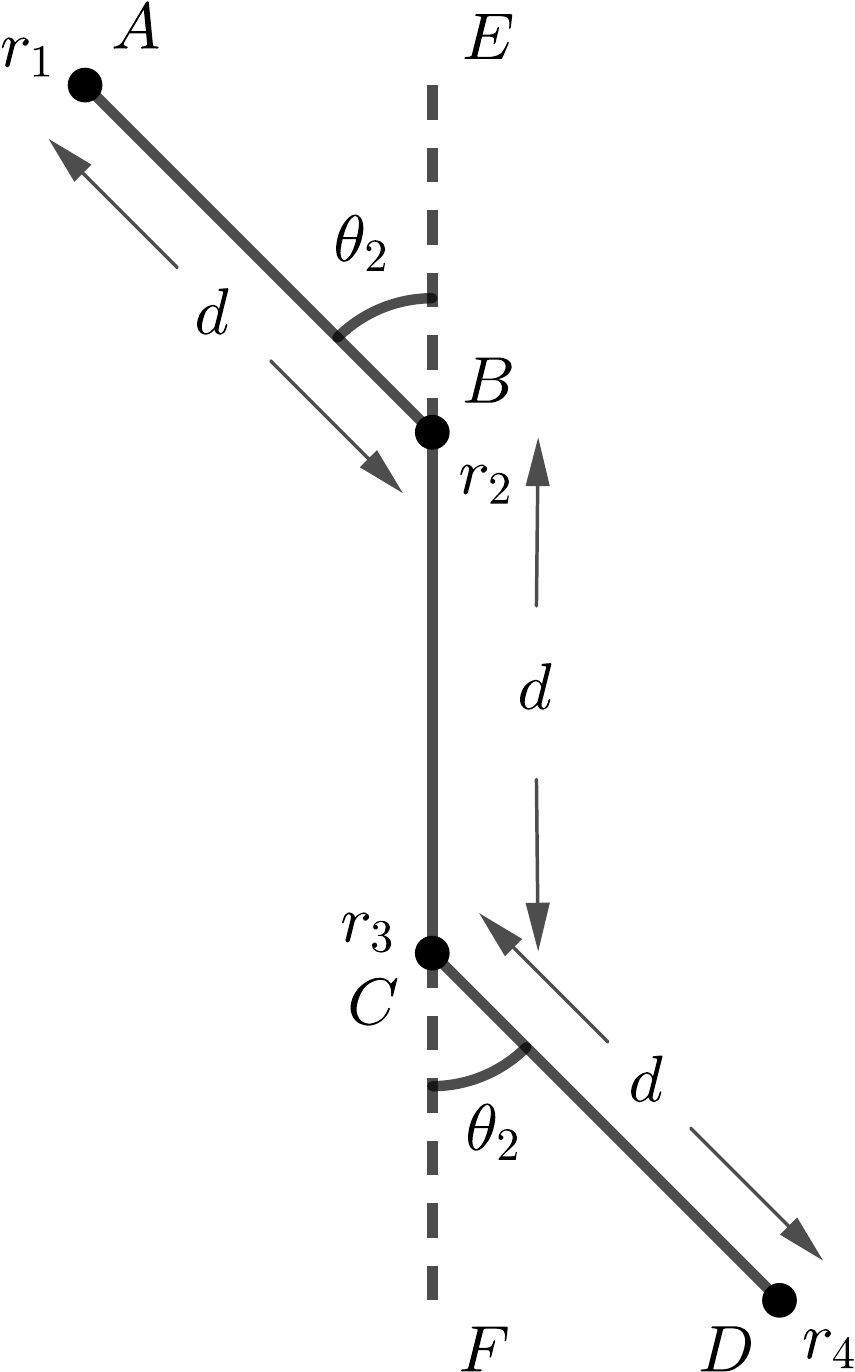}
\caption{Configuration (II) of Problem $AE$}
\label{fig:Iso1}
\end{minipage}
\end{figure}

\begin{figure}[h!]
\centering
\includegraphics[width=1.8cm]{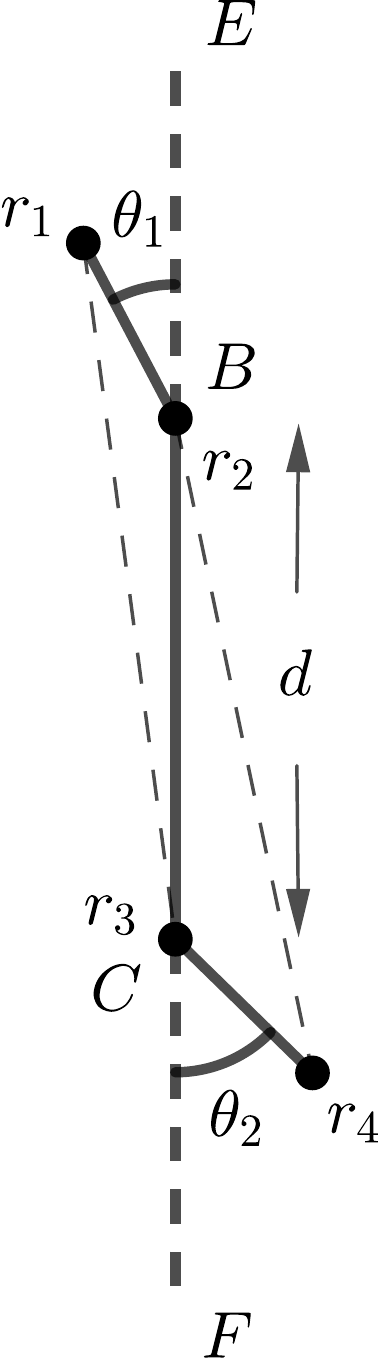}
\caption{Visibility Range Gap}
\label{fig:Iso1}
\end{figure}

Under full visibility conditions, each robot can see all the robot locations in the plane. Now each robot can uniquely identify its position in the plane. Therefore whenever the robot $r_1$ is activated, it will move to the position $A'$ such that the $\angle A'BE = \theta_2$. The rest of the robots will not move. After the robot $r_1$ moves all the robots can perceive that Configuration (II) is obtained and henceforward there will be no further movement of the robots. Hence the problem is solved.

\begin{lemma}\label{new lemma 1}
    $\forall$  $R\in \mathcal{R}_4$, $AE\in\mathcal{OBLOT_{F.V}^S}$.
\end{lemma}

\begin{corollary}\label{new cor 1}
     $\forall$  $R\in \mathcal{R}_4$, $AE\in\mathcal{OBLOT_{F.V}^F}$.
\end{corollary}

\begin{proof}
    Follows from Lemma \ref{new lemma 1}.
\end{proof}

\begin{corollary}\label{set 1 cor 2}
     $\forall$  $R\in \mathcal{R}_4$, $AE\in\mathcal{FSTA_{F.V}^S}$.
\end{corollary}

\begin{proof}
    Follows from Lemma \ref{new lemma 1}.
\end{proof}

\begin{corollary}\label{set 1 cor 4}
     $\forall$  $R\in \mathcal{R}_4$, $AE\in\mathcal{FSTA_{F.V}^F}$.
\end{corollary}

\begin{proof}
    Follows from Corollary \ref{new cor 1}.
\end{proof}

\begin{corollary}\label{set 1 cor 3}
     $\forall$  $R\in \mathcal{R}_4$, $AE\in\mathcal{FCOM_{F.V}^S}$.
\end{corollary}

\begin{proof}
    Follows from Lemma \ref{new lemma 1}.
\end{proof}

\begin{corollary}\label{set 1 cor 5}
     $\forall$  $R\in \mathcal{R}_4$, $AE\in\mathcal{FCOM_{F.V}^F}$.
\end{corollary}

\begin{proof}
    Follows from Corollary \ref{new cor 1}.
\end{proof}

\begin{corollary}\label{set 1 cor 6}
     $\forall$  $R\in \mathcal{R}_4$, $AE\in\mathcal{LUMI_{F.V}^S}$.
\end{corollary}

\begin{proof}
    Follows from Lemma \ref{new lemma 1}.
\end{proof}

\begin{corollary}\label{new Cor 7}
     $\forall$  $R\in \mathcal{R}_4$, $AE\in\mathcal{LUMI_{F.V}^F}$.
\end{corollary}

\begin{proof}
    Follows from Corollary \ref{new cor 1}.
\end{proof}
















\subsection{Impossibility of Solving $AE$ Problem in Limited Visibility Model}

\begin{lemma}\label{new lemma 2}

  $\exists$  $R\in \mathcal{R}_4$, $AE\not\in\mathcal{LUMI_{L.V.}^F}$.
\end{lemma}

\begin{proof}
    Let there exists an Algorithm $\mathcal{A}$ to solve the problem in $\mathcal{LUMI_{L.V.}^F}$. If Configuration (II) has to be formed from Configuration (I) then the robot $r_1$ must know the value of the angle it has to form. If $r_1$ has to move to the position $A'$ such that the $\angle A'BE = \theta_2$, then the robot $r_1$ must know the position of two robots $r_3$ and $r_4$ respectively in the initial configuration, i.e., the positions $C$ and $D$ respectively. Unless the position $C$ is known, the robot $r_1$ cannot perceive that it has to form the angle with respect to the extended line of the line segment. And unless it knows the position $D$, it cannot perceive the value $\theta_2$ that it has to form. But if $V_r = BC + \epsilon$, then it is not possible for the robot $r_1$ to see them from the initial configuration. Also according the requirement of the problem the robots $r_2$ and $r_3$ cannot move. Therefore to solve the problem $r_1$ must move. Now, if $r_1$ has to move, unless $r_1$ performs the required move to form Configuration (II) in one move, it has to move preserving the angle $\theta_1$. This is because $r_1$ does not know the fact that $\theta_1 < \theta_2$. The argument holds for $r_4$. And we have already seen that the initial configuration does not give the required information to form Configuration (II) in one move. 

    Now the only way the robot $r_1$ can move preserving the angle, is by moving along the line segment $AB$. Now note that if $r_1$ reaches $B$, the angle becomes $0$. Also as collisions are not allowed the robot $r_1$ cannot cross $B$. Similarly, the robot $r_3$ can only move along line segment $CD$ and it cannot cross the position $C$. Now, by moving along the line segments $AB$ and $CD$ respectively, however much the two robots $r_1$ and $r_4$ may come closer to $B$ and $C$ respectively, the adversary may choose the value of $\epsilon$ in such a way that the position $C$ is outside the visibility range of $r_1$ and $B$ is outside the visibility range of $r_2$. Note that the robots $r_2$ and $r_3$ cannot see $r_4$ and $r_1$ respectively, therefore it is also unknown to them which robot should perform the required move to form Configuration (II). Though $r_2$ and $r_3$ can measure the angles $\theta_1$ and $\theta_2$ respectively. It is not possible to store the value of the angles with finite number of lights. Hence the problem cannot be solved.
\end{proof}

From Lemma \ref{new lemma 2} follows:

\begin{corollary}\label{set 2 cor 1}
    $\exists$  $R\in \mathcal{R}_4$, $AE\not\in\mathcal{LUMI_{L.V.}^S}$
\end{corollary}

\begin{corollary}\label{set 2 cor 2}
    $\exists$  $R\in \mathcal{R}_4$, $AE\not\in\mathcal{FSTA_{L.V.}^F}$
\end{corollary}

\begin{corollary}\label{set 2 cor 4}
    $\exists$  $R\in \mathcal{R}_4$, $AE\not\in\mathcal{FSTA_{L.V.}^S}$
\end{corollary}

\begin{proof}
    Follows from Corollary \ref{set 2 cor 1}
\end{proof}

\begin{corollary}\label{set 2 cor 3}
    $\exists$  $R\in \mathcal{R}_4$, $AE\not\in\mathcal{FCOM_{L.V.}^F}$
\end{corollary}

\begin{corollary}\label{set 2 cor 5}
    $\exists$  $R\in \mathcal{R}_4$, $AE\not\in\mathcal{FCOM_{L.V.}^S}$
\end{corollary}

\begin{proof}
    Follows from Corollary \ref{set 2 cor 1}
\end{proof}

\begin{corollary}\label{set 2 cor 6}
    $\exists$  $R\in \mathcal{R}_4$, $AE\not\in\mathcal{OBLOT_{L.V.}^F}$
\end{corollary}

\begin{corollary}\label{set 2 cor 7}
    $\exists$  $R\in \mathcal{R}_4$, $AE\not\in\mathcal{OBLOT_{L.V.}^S}$
\end{corollary}

\begin{proof}
    Follows from Corollary \ref{set 2 cor 1}
\end{proof}











We get the following results:

\begin{theorem}
$\mathcal{OBLOT^F_{F.V.}> OBLOT^F_{L.V.}}$
\end{theorem}

\begin{proof}
   From  Corollary \ref{new cor 1} and Corollary \ref{set 2 cor 6}.
\end{proof}

\begin{theorem}
$\mathcal{OBLOT^S_{F.V.}> OBLOT^S_{L.V.}}$
\end{theorem}

\begin{proof}
    From  Lemma \ref{new lemma 1} and Corollary \ref{set 2 cor 7}.
\end{proof}

\begin{theorem}
$\mathcal{FSTA^F_{F.V.}> FSTA^F_{L.V.}}$
\end{theorem}

\begin{proof}
    From  Corollary \ref{set 1 cor 4} and Corollary \ref{set 2 cor 2}.
\end{proof}

\begin{theorem}
$\mathcal{FSTA^S_{F.V.}> FSTA^S_{L.V.}}$
\end{theorem}

\begin{proof}
    From  Corollary \ref{set 1 cor 2} and Corollary \ref{set 2 cor 4}.
\end{proof}

\begin{theorem}
$\mathcal{FCOM^F_{F.V.}> FCOM^F_{L.V.}}$
\end{theorem}

\begin{proof}
    From  Corollary \ref{set 1 cor 5} and Corollary \ref{set 2 cor 3}.
\end{proof}

\begin{theorem}
$\mathcal{FCOM^S_{F.V.}> FCOM^S_{L.V.}}$
\end{theorem}

\begin{proof}
    From  Corollary \ref{set 1 cor 3} and Corollary \ref{set 2 cor 5}.
\end{proof}

\begin{theorem}
$\mathcal{LUMI^F_{F.V.}> LUMI^F_{L.V.}}$
\end{theorem}

\begin{proof}
    From  Corollary \ref{new Cor 7} and Lemma \ref{new lemma 2}.
\end{proof}

\begin{theorem}
$\mathcal{LUMI^S_{F.V.}> LUMI^S_{L.V.}}$
\end{theorem}

\begin{proof}
    From  Corollary \ref{set 1 cor 6} and Corollary \ref{set 2 cor 1}.
\end{proof}

\section{Equivalent Oscillation Problem}\label{6}

\begin{definition}
    \textbf{Problem Equivalent Oscillation (EqOsc)}: Let three robots $r_1$, $r_2$ and $r_3$ be initially placed at three points $B$, $A$ and $C$ respectively. $AB=AC= d$. Let $B'$, $C'$ be points collinear on the line such that $AB' = AC' = \frac{2d}{3}$. The robots $r_1$ and $r_3$ have to change their positions from $B$ to $B'$ and back to $B$, $C$ to $C'$ and back to $C$ respectively, while always being equidistant from $r_2$, i.e., $A$ (Equidistant Condition). 
    
    More formally speaking, if there is a round $t$ such that the robots $r_1$ and $r_3$ is at $B$ and $C$ respectively, then there must must exist a round $t' > t$, such that $r_1$ and $r_3$ is at $B'$ and $C'$ respectively. Similarly if at round $t'$, $r_1$ and $r_3$ is at $B'$ and $C'$ respectively, then there must exist a round $t''>t'$, such that $r_1$ and $r_3$ is at $B$ and $C$ respectively (Oscillation Condition). 
\end{definition}

\begin{figure}[ht!]
\centering
\includegraphics[width=7cm]{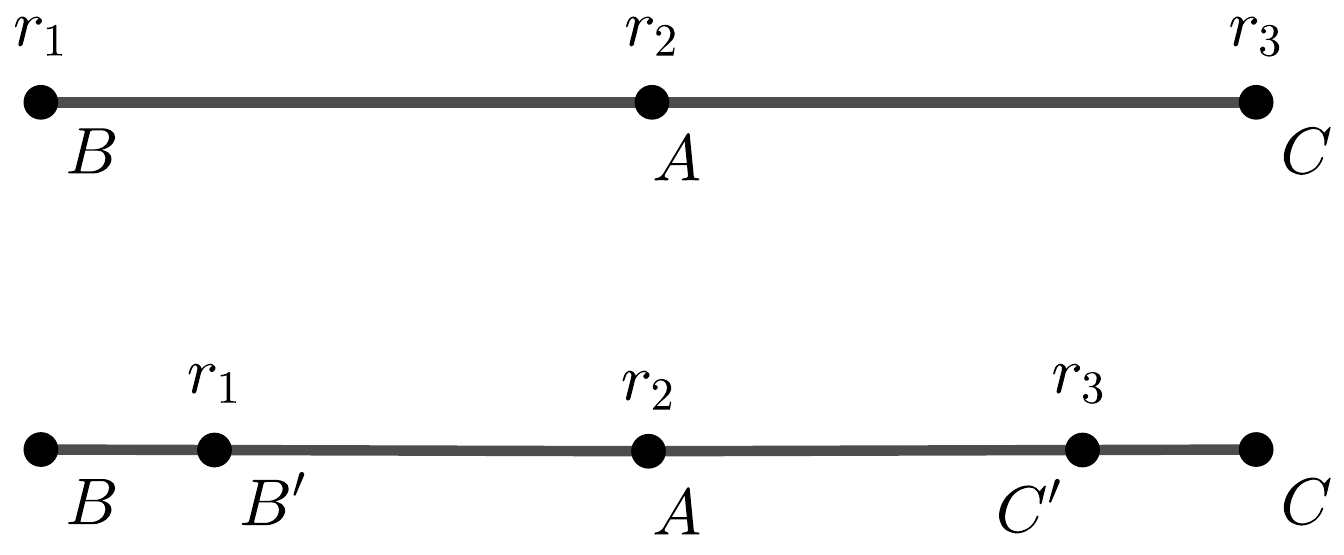}
\caption{Illustration of $EqOsc$ problem}
\label{fig:Iso1}
\end{figure}

We prove that this problem is not solvable in $\mathcal{LUMI_{F.V.}^S}$.


\begin{lemma}\label{lemma 7}
$\exists$  $R\in \mathcal{R}_3$, $EqOsc \not \in \mathcal{LUMI_{F.V.}^S}$.
\end{lemma}

\begin{proof}
    Let the the robots $r_1$, $r_3$ be able to successfully execute the movements satisfying the conditions of the problem till round $t$. Let at the beginning of round $t+1$ the robots $r_1$ and $r_3$ be at the points $B$ and $C$ respectively. So next the robots $r_1$ and $r_3$ must move to the points $B'$, $C'$ respectively. Note that the robots must move together or otherwise the equidistant condition is not satisfied. From round $t+1$ we activate only one of the terminal robots alternatively. Let at rounds $t+1$, $t+3$, $t+5$,......., the robot $r_1$ is activated and let at rounds $t+2$, $t+4$, $t+6$,......, the robot $r_3$ is activated.

Now whenever $r_1$ or $r_3$ makes a movement(when they are activated) the equidistant condition is violated. If neither $r_1$ nor $r_3$ makes any movement indefinitely then the oscillating condition is violated.

The problem cannot be solved in $\mathcal{LUMI_{F.V.}^S}$.

\end{proof}

From Lemma \ref{lemma 7} following result naturally follows,

\begin{corollary}\label{corollary 5}
$\exists$  $R\in \mathcal{R}_3$, $EqOsc \not \in \mathcal{FSTA_{F.V.}^S}$ and, $\exists$  $R\in \mathcal{R}_3$,  $EqOsc \not \in \mathcal{FCOM_{F.V.}^S}$.
\end{corollary}

\subsection{Solution of problem Equivalent Oscillation in $\mathcal{FSTA^{F}_{L.V.}}$}

We now give an algorithm to solve the problem in $\mathcal{FSTA_{L.V.}^F}$. The pseudocode of the algorithm is given below.

\begin{algorithm}[H]
\scriptsize
$d$ = distance from the closer robot\;
$A$ = Position of the middle robot\;

\eIf{not a terminal robot}
{Remain static}
{\eIf{Light = Off or Light = F}
     {Light $\leftarrow$ N
     
     Move to a point $D$ on the line segment such that $AD = \frac{2d}{3}$}
     {
        Light $\leftarrow$ F

         Move to a point $D$ on the line segment such that $AD = \frac{3d}{2}$

     }   
 }
 \caption{Algorithm $AlgEOSTA$ for Problem EqOsc executed by each robot $r$ in $\mathcal{FSTA_{L.V.}^F}$ }
\end{algorithm} 


     



\paragraph{Description and Correctness of $AlgEOSTA$}
Let the three robots $r_1$, $r_2$ and $r_3$ be at $B$, $A$, and $C$ respectively. Here $V_r> AB$, and $V_r>AC$, but $V_r< BC$. So there are three robots among which there is one robot which can see two other robots except itself, we call this robot the $middle$ $robot$. The other two robots can see only one other robot except itself. We call each of these two robots $terminal$ $robot$. The terminal robots are initially at a distance $D$ from the middle robot. Whenever a robot is activated it can understand whether it is a terminal or a middle robot. Now, each of the robots save a light which is initially saved to $Off$. If a robot perceives that it is a middle robot it does not do anything. If it is a terminal robot and its light is set to $Off$ or $F$, it changes its light to $N$ and moves closer a distance  two-third of the present distance from the middle robot, and if its light is set to $N$, it changes its light to $F$ and moves further to a distance $1.5$ times of the present distance from the middle robot. As we consider a fully synchronous system both the terminal robots execute the nearer and further movement alternatively together, and hence our problem is solved.

Hence we get the following result:

\begin{lemma}\label{lemma 8}
$\forall$  $R\in \mathcal{R}_3$, $EqOsc \in \mathcal{FSTA_{L.V.}^F}$.
\end{lemma}


\begin{theorem}\label{theorem 9}
    $\mathcal{FSTA_{L.V.}^F} > \mathcal{FSTA_{L.V.}^S}$
\end{theorem}

\begin{proof}
    By Corollary \ref{corollary 5} the problem cannot be solved in $\mathcal{FSTA_{L.V.}^S}$, and, trivially  $\mathcal{FSTA_{L.V.}^F} \geq \mathcal{FSTA_{L.V.}^S}$.

\end{proof}

\begin{theorem}\label{theorem 10}
$\mathcal{FSTA_{L.V.}^F} \perp \mathcal{FSTA_{F.V.}^S}$
\end{theorem}

\begin{proof}

By Corollary \ref{corollary 5} and Lemma \ref{lemma 8}, $EqOsc$ cannot be solved in $\mathcal{FSTA_{F.V.}^S}$ but can be solved in $\mathcal{FSTA_{L.V.}^F}$.

    Similarly, by Corollary \ref{set 2 cor 2} and \ref{set 1 cor 2}  $AE$ cannot be solved in $\mathcal{FSTA_{L.V.}^F}$ but can be solved in $\mathcal{FSTA_{F.V.}^S}$. Hence the result.


\end{proof}



\subsection{Solution of problem Equivalent Oscillation in $\mathcal{FCOM^{F}_{L.V.}}$}

We now give an algorithm to solve the problem in $\mathcal{FCOM_{L.V.}^F}$. The pseudocode of the algorithm is given below.

\begin{algorithm}[H]
\scriptsize
$d$ = distance from the closer robot\;
$A$ = Position of the middle robot\;

\eIf{not a terminal robot}
{Remain static

  \If{Visible light = $NIL$ or $FAR$}
       {
       Light $\leftarrow$ $FAR$
       }

      \ElseIf{Visible light = $NEAR$} {

       Light $\leftarrow$ $NEAR$

       }

  }
{\If{Visible light = $NIL$ or Visible light = $NEAR$}
     {Light $\leftarrow$ $NEAR$

     Move to a point $D$ on the line segment such that $AD = \frac{2d}{3}$
     
     }
     
     \ElseIf{Visible light = $FAR$}{
        Light $\leftarrow$ $FAR$
        
        Move to a point $D$ on the line segment such that $AD = \frac{3d}{2}$

     }   
 }
 \caption{Algorithm $AlgEOCOM$ for Problem EqOsc executed by each robot $r$ in $\mathcal{FCOM_{L.V.}^F}$ }
\end{algorithm}








     
     



\paragraph{Description and Correctness of $AlgEOCOM$}
Let the three robots $r_1$, $r_2$ and $r_3$ be at $B$, $A$, and $C$ respectively. Here $V_r> AB$, and $V_r>AC$, but $V_r< BC$. So there are three robots among which there is one robot which can see two other robots except itself, we call this robot the $middle$ $robot$. The other two robots can see only one other robot except itself. We call each of these two robots $terminal$ $robot$.  Whenever a robot is activated it can understand whether it is a terminal or a middle robot. Now, each of the robots has a light which is initially saved to $NIL$. If a robot perceives that it is a middle robot it can see the lights of the two terminal robots. If it perceives the lights of the terminal robots to be set to $NIL$ or $FAR$, it sets its own light to $FAR$. And if it perceives the lights of the terminal robots to be set to $NEAR$, it sets its own light to $NEAR$.  The middle robot only changes its light but does not change its position.
If it is a terminal robot it can only see the light of the middle robot. If the light of the middle robot is set to $NIL$ or $NEAR$, the terminal robot  changes its light to $NEAR$ and moves closer to a distance which is two-third of the present distance from the middle robot. And if the light of the middle robot is set to $FAR$, it changes its light to $FAR$ and moves away to a distance $1.5$ times the present distance from the middle robot. As we consider a fully synchronous system both the terminal robots execute the nearer and further movement alternatively together, and hence our problem is solved.

Hence we get the following result:

\begin{lemma}\label{lemma 9}
$\forall$  $R\in \mathcal{R}_3$, $EqOsc \in \mathcal{FCOM_{L.V.}^F}$.
\end{lemma}


\begin{theorem}\label{theorem 11}
    $\mathcal{FCOM_{L.V.}^F} > \mathcal{FCOM_{L.V.}^S}$
\end{theorem}

\begin{proof}
 By Corollary \ref{corollary 5} the problem cannot be solved in $\mathcal{FCOM_{L.V.}^S}$, and, trivially  $\mathcal{FCOM_{L.V.}^F} \geq \mathcal{FCOM_{L.V.}^S}$. Hence our theorem.   
\end{proof}

\begin{theorem}\label{theorem 12}
$\mathcal{FCOM_{L.V.}^F} \perp \mathcal{FCOM_{F.V.}^S}$
\end{theorem}

\begin{proof}

 By Corollary \ref{corollary 5} and Lemma \ref{lemma 9}, $EqOsc$ cannot be solved in $\mathcal{FCOM_{F.V.}^S}$ but can be solved in $\mathcal{FCOM_{L.V.}^F}$.

    Similarly, by Corollary \ref{set 2 cor 3} and \ref{set 1 cor 3}  $AE$ cannot be solved in $\mathcal{FCOM_{L.V.}^F}$ but can be solved in $\mathcal{FCOM_{F.V.}^S}$. Hence the result.


\end{proof}

\subsection{Similar deductions in $\mathcal{OBLOT}$ and $\mathcal{LUMI}$}

\begin{theorem}\label{theorem 3}
    $\mathcal{OBLOT_{L.V.}^F} \perp \mathcal{OBLOT_{F.V.}^S}$
\end{theorem}

\begin{proof}
     We have proved that the problem $AE$ cannot be solved in $\mathcal{OBLOT_{L.V.}^F}$ but can be solved in $\mathcal{OBLOT_{F.V.}^S}$. 

    Now we consider the problem $Rendezvous$ which was proved to be impossible in $\mathcal{OBLOT_{F.V.}^S}$ in \cite{SuzukiY99}. Now we claim that in our model the problem is possible to solve in $\mathcal{OBLOT_{L.V.}^F}$. This is because in our model we assume the visibility graph of the robots in the initial configuration to be connected. Now when there are only two robots in the initial configuration this means, all the robots in the initial configuration can see each other. Hence in this case the problem reduces to $\mathcal{OBLOT^F}$ model. Now it is well known that $rendezvous$ problem is solvable in this model, as the robots just move to the mid-point of the line segment joining them. Hence the result.

\end{proof}

\begin{theorem}\label{theorem 4}
    $\mathcal{OBLOT_{L.V}^F} > \mathcal{OBLOT_{L.V.}^S}$
\end{theorem}

\begin{proof}
     Trivially $\mathcal{OBLOT_{L.V}^F} \geq \mathcal{OBLOT_{L.V.}^S}$ and by Theorem \ref{theorem 3} the problem $Rendezvous$ is solvable in $\mathcal{OBLOT_{L.V}^F}$ but not in $\mathcal{OBLOT_{L.V.}^S}$.
 
\end{proof}

\begin{lemma}\label{new lemma 6}
$\forall$  $R\in \mathcal{R}_3$, $EqOsc \in \mathcal{LUMI_{L.V.}^F}$.
\end{lemma}

\begin{proof}
    From Lemma \ref{lemma 8} and \ref{lemma 9}.
\end{proof}

\begin{theorem}
    $\mathcal{LUMI_{L.V.}^F} \perp \mathcal{LUMI_{F.V.}^S}$
\end{theorem}

\begin{proof}
     By Lemma \ref{new lemma 2}   and Corollary \ref{new Cor 7}  the problem $AE$ cannot be solved in $\mathcal{LUMI_{L.V.}^F}$ but can be solved in $\mathcal{LUMI_{F.V.}^S}$. 

     By Lemma \ref{lemma 7} and \ref{new lemma 6} , the problem $EqOsc$ cannot be solved in $\mathcal{LUMI_{F.V.}^S}$ but can be solved in $\mathcal{LUMI_{L.V.}^F}$.
     Hence, the result.


\end{proof}

\section{Conclusion}\label{7}

In this paper we have initiated the analysis of computational capabilities of mobile robots having limited visibility. We have considered all the four models $\{ \mathcal{OBLOT}, \mathcal{FSTA}, \mathcal{FCOM}, \mathcal{LUMI} \}$, we have more or less exhaustively analyzed the possible relationships when the model is fixed. But a huge amount of interesting questions still remains to be solved. The possible future directions are:
\begin{enumerate}
    
    \item Cross-model relationships under limited visibility conditions
    \item The computational relationships when the scheduler is asynchronous.
\end{enumerate}


\bibliographystyle{splncs04}
\bibliography{samplepaper}

\begin{thebibliography}{10}
\providecommand{\url}[1]{\texttt{#1}}
\providecommand{\urlprefix}{URL }
\providecommand{\doi}[1]{https://doi.org/#1}

\bibitem{AndoOSY99}
Ando, H., Oasa, Y., Suzuki, I., Yamashita, M.: Distributed memoryless point convergence algorithm for mobile robots with limited visibility. {IEEE} Trans. Robotics Autom.  \textbf{15}(5),  818--828 (1999). \doi{10.1109/70.795787}, \url{https://doi.org/10.1109/70.795787}

\bibitem{BuchinFKPSW21}
Buchin, K., Flocchini, P., Kostitsyna, I., Peters, T., Santoro, N., Wada, K.: Autonomous mobile robots: Refining the computational landscape. In: {IEEE} International Parallel and Distributed Processing Symposium Workshops, {IPDPS} Workshops 2021, Portland, OR, USA, June 17-21, 2021. pp. 576--585. {IEEE} (2021). \doi{10.1109/IPDPSW52791.2021.00091}, \url{https://doi.org/10.1109/IPDPSW52791.2021.00091}

\bibitem{BuchinFKPSW22}
Buchin, K., Flocchini, P., Kostitsyna, I., Peters, T., Santoro, N., Wada, K.: On the computational power of energy-constrained mobile robots: Algorithms and cross-model analysis. In: Parter, M. (ed.) Structural Information and Communication Complexity - 29th International Colloquium, {SIROCCO} 2022, Paderborn, Germany, June 27-29, 2022, Proceedings. Lecture Notes in Computer Science, vol. 13298, pp. 42--61. Springer (2022). \doi{10.1007/978-3-031-09993-9\_3}, \url{https://doi.org/10.1007/978-3-031-09993-9\_3}

\bibitem{DasFPSY12}
Das, S., Flocchini, P., Prencipe, G., Santoro, N., Yamashita, M.: The power of lights: Synchronizing asynchronous robots using visible bits. In: 2012 {IEEE} 32nd International Conference on Distributed Computing Systems, Macau, China, June 18-21, 2012. pp. 506--515. {IEEE} Computer Society (2012). \doi{10.1109/ICDCS.2012.71}, \url{https://doi.org/10.1109/ICDCS.2012.71}

\bibitem{0001FPSY16}
Das, S., Flocchini, P., Prencipe, G., Santoro, N., Yamashita, M.: Autonomous mobile robots with lights. Theor. Comput. Sci.  \textbf{609},  171--184 (2016). \doi{10.1016/j.tcs.2015.09.018}, \url{https://doi.org/10.1016/j.tcs.2015.09.018}

\bibitem{FlocchiniPSW05}
Flocchini, P., Prencipe, G., Santoro, N., Widmayer, P.: Gathering of asynchronous robots with limited visibility. Theor. Comput. Sci.  \textbf{337}(1-3),  147--168 (2005). \doi{10.1016/j.tcs.2005.01.001}, \url{https://doi.org/10.1016/j.tcs.2005.01.001}

\bibitem{FlocchiniSVY16}
Flocchini, P., Santoro, N., Viglietta, G., Yamashita, M.: Rendezvous with constant memory. Theor. Comput. Sci.  \textbf{621},  57--72 (2016). \doi{10.1016/j.tcs.2016.01.025}, \url{https://doi.org/10.1016/j.tcs.2016.01.025}

\bibitem{FlocchiniSW19}
Flocchini, P., Santoro, N., Wada, K.: On memory, communication, and synchronous schedulers when moving and computing. In: Felber, P., Friedman, R., Gilbert, S., Miller, A. (eds.) 23rd International Conference on Principles of Distributed Systems, {OPODIS} 2019, December 17-19, 2019, Neuch{\^{a}}tel, Switzerland. LIPIcs, vol.~153, pp. 25:1--25:17. Schloss Dagstuhl - Leibniz-Zentrum f{\"{u}}r Informatik (2019). \doi{10.4230/LIPIcs.OPODIS.2019.25}, \url{https://doi.org/10.4230/LIPIcs.OPODIS.2019.25}

\bibitem{GoswamiSGS22}
Goswami, P., Sharma, A., Ghosh, S., Sau, B.: Time optimal gathering of myopic robots on an infinite triangular grid. In: Devismes, S., Petit, F., Altisen, K., Luna, G.A.D., Anta, A.F. (eds.) Stabilization, Safety, and Security of Distributed Systems - 24th International Symposium, {SSS} 2022, Clermont-Ferrand, France, November 15-17, 2022, Proceedings. Lecture Notes in Computer Science, vol. 13751, pp. 270--284. Springer (2022). \doi{10.1007/978-3-031-21017-4\_18}, \url{https://doi.org/10.1007/978-3-031-21017-4\_18}

\bibitem{PoudelS21}
Poudel, P., Sharma, G.: Time-optimal gathering under limited visibility with one-axis agreement. Inf.  \textbf{12}(11), ~448 (2021). \doi{10.3390/info12110448}, \url{https://doi.org/10.3390/info12110448}

\bibitem{SuzukiY99}
Suzuki, I., Yamashita, M.: Distributed anonymous mobile robots: Formation of geometric patterns. {SIAM} J. Comput.  \textbf{28}(4),  1347--1363 (1999). \doi{10.1137/S009753979628292X}, \url{https://doi.org/10.1137/S009753979628292X}

\end{thebibliography}

\end{document}